\documentclass[letterpaper, 10 pt, conference]{ieeeconf}
\IEEEoverridecommandlockouts                              

\usepackage[utf8]{inputenc}
\usepackage[english]{babel}
\usepackage{cite}
\usepackage{graphicx}
 \graphicspath{{./figures/}}
\usepackage[justification=centering]{caption}
\usepackage{subcaption}
\usepackage{mathtools}
\usepackage{amsmath}

\usepackage{amsthm,amsfonts,amstext,amssymb}
\usepackage{dsfont}
\let\labelindent\relax
\usepackage{enumitem}   
\usepackage{indentfirst}
\usepackage[font=small,skip=0pt]{caption}
\usepackage{float}
\usepackage{xcolor}
\usepackage{xr}
\usepackage{tikz}
\usetikzlibrary{shapes,arrows,positioning,calc}
\usepackage{standalone}
\usepackage[pdftex]{hyperref}
\usepackage[normalem]{ulem}
\usepackage{algorithm}
\usepackage{algpseudocode}

\newtheorem{theorem}{Theorem}

\theoremstyle{definition}
\newtheorem{definition}{Definition}
\theoremstyle{remark}
\newtheorem{remark}{Remark}

\algrenewcommand\algorithmicrequire{\textbf{Input:}}
\newcommand{\comment}[1]{}

\newcommand{\ra}{\rightarrow}

\newcommand{\eps}{\varepsilon}

\newcommand{\N}{\mathbb{N}}
\newcommand{\E}[1]{\mathbb{E}\left[ #1 \right]}

\newcommand{\indic}[1]{\mathds{1}_{\{#1\}}}

\newcommand{\vehiclelength}{L}
\newcommand{\timeheadway}{h}
\newcommand{\standstilldist}{S_0}
\newcommand{\speedlim}{V_f}

\newcommand{\minaccel}{a_{\min}}

\newcommand{\Speedmode}{\text{speed tracking}}
\newcommand{\Safetymode}{\text{safety}}

\newcommand{\StateofDySys}{X}

\newcommand{\Agent}[1]{Vehicle{#1}}
\newcommand{\agent}[1]{vehicle{#1}}
\newcommand{\source}[1]{on-ramp{#1}}
\newcommand{\sink}[1]{off-ramp{#1}}
\newcommand{\route}[1]{route{#1}}

\newcommand{\numsources}{\mathcal{M}}

\newcommand{\numpaths}{\mathcal{P}}
\newcommand{\numcells}[1]{r_{#1}}

\newcommand{\timestep}{\tau}
\newcommand{\routingmatrix}{R}

\newcommand{\Arrivalrate}[1]{\lambda_{#1}}
\newcommand{\Numarrivals}[1]{A_{#1}}

\newcommand{\Numdepartures}[1]{D_{#1}}
\newcommand{\Cumuldepartures}[1]{\bar{D}_{#1}}

\newcommand{\Queuelength}[1]{Q_{#1}}
\newcommand{\Nodedegree}[1]{N_{#1}}

\newcommand{\Avgload}[1]{\rho_{#1}}
\newcommand{\Nodedest}[1]{M_{#1}}
\newcommand{\Lyap}[1]{V\left( #1 \right)}
\newcommand{\Cyclelength}{T}
\newcommand{\StateofMC}{Y}
\newcommand{\StateofMCExp}{Z}

\newcommand{\StateSpace}{\mathcal{S}}
\newcommand{\MaxNumArrival}[3]{\tilde{A}_{#1}^{#2}(#3)}

\newcommand{\DesDecreaseconst}{\gamma_{1}}
\newcommand{\ReleasetimeInc}[1]{\theta_{#1}}

\newcommand{\ReleasetimeDecConst}[1]{\gamma_{#1}}
\newcommand{\Releasetime}[1]{g_{#1}}
\newcommand{\ReleasetimeAdj}[1]{\beta_{#1}}

\newcommand{\Tempty}{T_{\text{free}}}
\newcommand{\updateperiod}{T_{\text{per}}}

\newcommand{\DRR}{\text{DRRA}}

\begin{document}

\title{\LARGE \bf
Throughput of Freeway Networks under Ramp Metering Subject to Vehicle Safety Constraints}
\author{Milad~Pooladsanj$^{1}$,
        Ketan~Savla$^{2}$,
        and~Petros~A.~Ioannou$^{1}$
\thanks{$^{1}$M. Pooladsanj and P. A. Ioannou are with the Department
of Electrical Engineering, University of Southern California, Los Angeles, CA 90007 USA {\tt\small pooladsa@usc.edu; ioannou@usc.edu}.}
\thanks{$^{2}$K. Savla is with the Sonny Astani Department of Civil and Environmental Engineering, University of Southern California, Los Angeles CA 90089 USA {\tt\small ksavla@usc.edu}. K. Savla has financial interest in Xtelligent, Inc.} 
\thanks{
This work was supported in part by NSF CMMI 1636377 and METRANS 19-17.}
}\maketitle
\begin{abstract}
Ramp metering is one of the most effective tools to combat traffic congestion. In this paper, we present a ramp metering policy for a network of freeways with arbitrary number of on- and off-ramps, merge, and diverge junctions. The proposed policy is designed at the microscopic level and takes into account vehicle following safety constraints. In addition, each on-ramp operates in cycles during which it releases vehicles as long as the number of releases does not exceed its queue size at the start of the cycle. Moreover, each on-ramp dynamically adjusts its release rate based on the traffic condition. To evaluate the performance of the policy, we analyze its throughput, which is characterized by the set of arrival rates for which the queue sizes at all on-ramps remain bounded in expectation. We show that the proposed policy is able to maximize the throughput if the merging speed at all the on-ramps is equal to the free flow speed and the network has no merge junction. We provide simulations to illustrate the performance of our policy and compare it with a well-known policy from the literature.

\end{abstract}
\section{Introduction}\label{sec:introduction}
Ramp metering, which controls the inflow of traffic to a freeway, is one of the most efficient tools to combat traffic congestion. In this paper, we propose a microscopic (vehicle)-level ramp metering policy for freeway networks with arbitrary number of on- and off-ramps, merge and diverge junctions. Our approach builds up on our previous work \cite{pooladsanj2022saturation}, where we considered a single freeway with no merge or diverge junctions. 
\par
\comment{
The ramp metering design problem is often done for a stretch of a freeway, with no freeway-to-freeway merge or diverge, and by using macroscopic traffic flow models. 
First paragraph: 
\begin{itemize}
        Control the inflow for a general freeway network consists of merge and diverge.
    \item Why is it important?
    \begin{itemize}
        \item Merging, either through on-ramps or other freeways, is a main source of congestion
    \end{itemize}
    \item How can we address it?
    \begin{itemize}
        \item Can be taken into account in RM design 
        \item We design RM at the microscopic level
    \end{itemize}
\end{itemize}
}
Ramp metering policies are commonly designed using macroscopic traffic flow models, which use spatio-temporal averaging of vehicle interactions to model the traffic. Many ramp metering policies have been proposed in the literature using this approach \cite{papageorgiou2002freeway, papageorgiou2003review}. These policies are either: (i) fixed-time \cite{wattleworth1965peak}, where no real-time traffic measurement is used, or (ii) local traffic-responsive \cite{papageorgiou1991alinea}, where on-ramps use only local measurements, or (iii) coordinated traffic-responsive \cite{papageorgiou1990modelling, papamichail2010coordinated, gomes2006optimal}, where on-ramps also use measurements from other parts of the network. In spite of its generality, the macroscopic-level approach is not suitable in the context of Connected and Automated Vehicles (CAVs) as it does not capture some of their capabilities, such as the ability to dissipate traffic disturbance and provide accurate measurements of traffic through Vehicle-to-Vehicle (V2V) and Vehicle-to-Infrastructure (V2I) communication systems \cite{stern2018dissipation,pooladsanj2020vehicle,pooladsanj2021vehicle}. Additionally, macroscopic-level ramp metering policies do not account for vehicle following safety constraints. As a consequence, the vehicles may find themselves in unsafe situations near an on-ramp, which can trigger congestion or even an accident. While adding safety filters on top of the ramp metering policy is a potential solution, it may adversely affect other performance criteria such as the total travel time \cite{pooladsanj2022saturation}. 
\par 
A more systematic approach is to explicitly include vehicle safety, connectivity, and automation protocols in the ramp metering design by considering the vehicles at the microscopic level. Somewhat surprisingly, the literature on the microscopic-level ramp metering design is sparse and only limited to isolated on-ramps \cite{rios2016automated, Malikopoulos.2017}. Relatively little attention has been given to analyzing the overall performance of a freeway network. For example, are certain system-level performance metrics such as \emph{throughput} optimized? To bridge this gap and expand our previous work for a single freeway \cite{pooladsanj2022saturation}, we consider a network of freeways to design and analyze a microscopic ramp metering policy that can maximize the throughput.  
\comment{
\begin{itemize}
    \item RM design often done at the macroscopic level
    \item Requires traffic flow model which is challenging for general networks, especially because of merge
    \item it is possible to do for individual freeway segment with no merge or diverge [reference], but this may lead to suboptimal policies. 
    \item For example, it could create congestion at a merge junction, spillback to individual incoming freeway segments
    \item Cannot capture the impact of CAVs
\end{itemize}

Third paragraph: more literature review
\begin{itemize}
    \item More natural to do at the microscopic level, especially because it can capture impact of CAVs, connectivity may be explicitly taken into account
    \item There is a large gap in the literature. Most focus on designing at a single on-ramp level. Could be inefficient for the same reasons mentioned before.
    \item Motivated by our previous work which concerned a single freeway segment, we extend it to general networks 
\end{itemize}
}
\par
The freeway network is modeled as a directed acyclic graph, with arbitrary number of on- and off-ramps, merge, and diverge junctions. \Agent{s} arrive from outside the network to the \source{s} and join the \source{} queues. The \source{} arrivals are sampled from i.i.d Bernoulli processes that are independent across different \source{s}, and the vehicle routes are sampled independently from a routing matrix. Once released from an \source{}, each \agent{} follows standard safety and speed rules until it reaches its destination \sink{}. \Agent{s} in the network have the same acceleration and braking capabilities, and are equipped with V2V and V2I communication systems. We leverage these connectivity capabilities to design a coordinated traffic-responsive ramp metering policy. The proposed policy operates under vehicle following safety constraints, where new vehicles are released only if there is enough gap on the mainline. Additionally, the on-ramps operate under synchronous cycles during which each \source{} does not release more \agent{s} than its queue size at the start of the cycle. Moreover, an \source{}'s release rate is adjusted during a cycle according to the traffic condition in the network. 
\par
We evaluate the performance of a policy in terms of its throughput, which is characterized by the set of arrival rates under which the queue sizes at every \source{} remain bounded in expectation. We show that the proposed policy is able to maximize the throughput when the merging speed at all the on-ramps equals the free flow speed and the network has no merge junction, e.g., a single freeway.
\par
The remainder of the paper is organized as follows: in Section~\ref{Section: Problem Formulation}, we describe the network structure, the \agent{s} behavior in the network, the demand model, and the performance metric. We provide our ramp metering policy and its performance analysis in Section~\ref{sec:results}. We provide simulation results in Section~\ref{sec:simulation}, and conclude the paper in Section~\ref{sec:conclusion}.
\par
We collect key notations to be used throughout the paper. Let $\N$ and $\N_{0}$ respectively denote the set of positive integers, and non-negative integers. For $m \in \N$, $[m]$ denotes the set $\{1, \ldots, m\}$. For a set $A$, $|A|$ denotes its cardinality.
\comment{
It is expected that CAVs will improve the traffic efficiency besides providing safety and comfort. However, the path toward achieving this goal is not clear \cite{wang2022ego}. 
\par
V2X communication is the key element which makes the designed control strategies efficient. These strategies could have been inefficient otherwise due to e.g., safety issues. 
\par
Our paper concerns the decision making process... . In particular, we do not design vehicle controllers that achieve a certain goal. We rather assume such controllers exist... .
}

\section{Problem Formulation}\label{Section: Problem Formulation}

\subsection{Network Topology}\label{sec:network-topology}
\begin{figure}[htb!]
    \centering
    \includegraphics[width=0.38\textwidth]{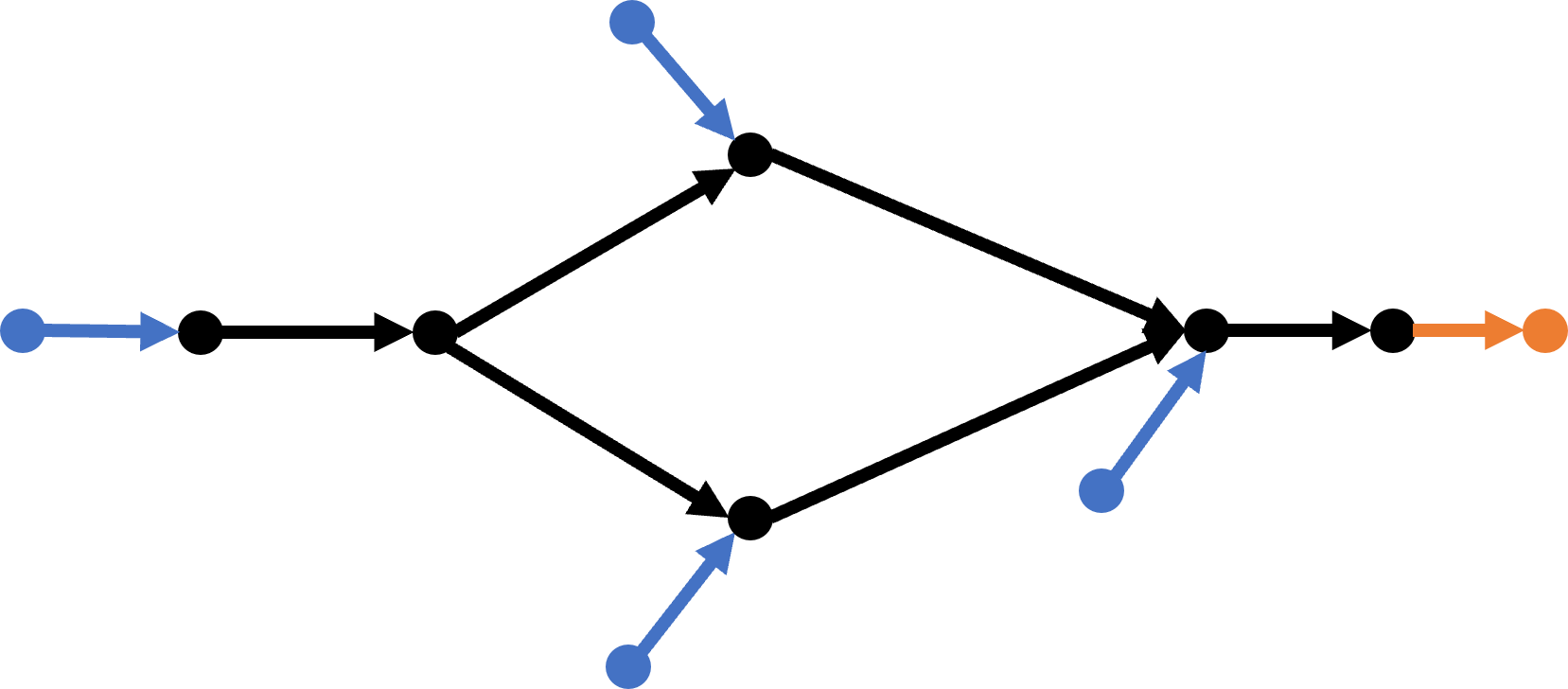}
    
    \vspace{0.1 cm}
    
    \caption{\sf An acyclic network with $4$ \source{s} (blue), $1$ \sink{} (orange), $6$ freeway segments, $1$ merge node, and $1$ diverge node.}
    \label{fig:network-example}
\end{figure}
The network is modeled as a directed acyclic graph. 
The nodes represent either junctions (merge or diverge), or sources/sinks; the edges represent either single-lane freeway segments that connect two junctions, or on- and \sink{s} that connect sources to merge nodes and diverge nodes to sinks, respectively; see Figure~\ref{fig:network-example}. We refer to any subset of the freeway segments as the mainline. For simplicity, we assume that each source has exactly one outgoing \source{}, and each sink has exactly one incoming \sink{}. We distinguish between two types of merge node: \source{} nodes, which have at least one incoming on-ramp, and regular merge nodes (or simply merge nodes), which have no incoming on-ramp. Similarly, we distinguish between an \sink{} node and a diverge node. We denote the set of \source{} nodes by $\numsources$. For simplicity, we assume that each \source{} node has exactly one incoming \source{} edge. Thus, $\numsources$ also specifies the set of \source{} edges. For \source{} $i \in \numsources$, we let $\numsources_{i}^{+}$ be the set of its successor \source{s}, i.e., the set of all \source{s} that are immediately downstream of \source{} $i$. An element of $\numsources_{i}^{+}$ is denoted by $i^{+}$. Similarly, we let $\numsources_{i}^{-}$ be the set of its predecessor \source{s}. We say that two edges are adjacent if the head of one edge is the tail of the other edge. A \emph{\route{}} is a sequence of adjacent edges, with the first edge being an \source{} and the last edge being an \sink{}. The collection of all \route{s} is denoted by $\numpaths$. 
\par
\Agent{s} arrive to \source{s} according to some exogenous stochastic process and join the \source{} queues. For any \source{}, we assume a point queue model with unlimited capacity. We also assume that \agent{s} that reach their destination \sink{} node, immediately leave the network without creating any obstruction. Once a vehicles is released from an on-ramp into the mainline, it follows standard speed and safety rules, which will be explained in the following section. 


\subsection{\Agent{}-Level Objectives}\label{sec:vehicle-control}
\Agent{s} have the same length $\vehiclelength$, the same acceleration and braking capabilities, and are equipped with V2V and V2I communication systems. The exact information communicated by a \agent{} will be specified in Section~\ref{sec:results}, where we introduce our ramp metering policy. We use the term \emph{ego} \agent{} to refer to a specific \agent{} under consideration, and denote it by $e$. Consider a \agent{} following scenario and let $v_e$ (resp. $v_l$) be the speed of the ego \agent{} (resp. its leading \agent{}), and $S_e$ be the \emph{safety distance} between the two \agent{s} required to avoid collision. We assume that $S_e$ satisfies
    $S_e = \timeheadway v_e + \standstilldist + \frac{v^2_{e} - v^2_l}{2|\minaccel|}$,
which is calculated based on an emergency stopping scenario with details given in \cite{Ioannou.Chien.1993}. Here, $\timeheadway > 0$ is a safe time headway constant, $\standstilldist > 0$ is an additional constant gap, and $\minaccel < 0$ is the minimum possible deceleration of the leading \agent{}. We consider two general modes of operation for each \agent{}: the \emph{$\Speedmode$} mode and the \emph{$\Safetymode$} mode. The main objective in the $\Speedmode$ mode is to adjust and maintain the speed to the \emph{free flow speed} $\speedlim$ when the ego \agent{} is far from any leading \agent{}; the main objective in the $\Safetymode$ mode is to avoid rear-end collision once the ego \agent{} gets close to the leading \agent{}.
\par
We assume that \agent{s} obey the following rules:
\begin{itemize}
\item[(VC1)] The ego \agent{} maintains the constant speed $\speedlim$ if:
\begin{itemize}
    \item[(a)] it is at a safe distance with respect to its leading \agent{} \footnote{At a diverge node, the leading \agent{} is uniquely determined by the ego \agent{}'s route choice.}, i.e., $y_e \geq S_e$, where $y_e$ is the distance between the two \agent{s}. Note that if its leading \agent{} is also at the constant speed $\speedlim$, then $y_e \geq S_e = \timeheadway \speedlim + \standstilldist$. This distance is equivalent to a time headway of at least $\timestep := \timeheadway + (\standstilldist + \vehiclelength)/\speedlim$ between the front bumpers of the two \agent{s}.
    \item[(b)] it is near a merge node and \emph{predicts} to be at a safe distance with respect to its leading \agent{} after reaching the merging node, i.e.,  $\hat{y}_e(t_m) \geq \hat{S}_e(t_m) := \timeheadway \hat{v}_e(t_m) + \standstilldist + \frac{\hat{v}_e^2(t_m) - \hat{v}_{\hat{l}}^2(t_m)}{2|\minaccel|}$, where $t_m$ is the time the ego \agent{} predicts to reach the merge node, $\hat{y}_e$ is the predicted distance between the two \agent{s}, and $\hat{v}$ is the predicted speed. 
    Each vehicle uses the following speed prediction rule to calculate its predicted time of merging, position, and speed: if in the $\Speedmode$ mode, it assumes that it remains in this mode in the future; if in the $\Safetymode$ mode, it assumes a constant speed trajectory. These information are then communicated among all \agent{s} near a merge node via V2V communication. 
\end{itemize}


\item[(VC2)] there exists $\Tempty$ such that for any initial condition, \agent{s} reach the \emph{free flow equilibrium} after at most $\Tempty$ time if no other \agent{} is released from the \source{s}. The free flow equilibrium refers to a state where each \agent{} is moving at the constant speed $\speedlim$ and will maintain this speed in the future if no additional \agent{} is released from the \source{s}.
\end{itemize}
\begin{remark}\label{remark:simplifying-assumption}
Due to lack of space, we assume that the ego \agent{} immediately joins the mainline after being released. In other words, we ignore the \agent{} dynamics between release from an \source{} and joining the mainline. We also assume the speed at which the ego \agent{} joins the mainline, i.e., the merging speed, is $\speedlim$. The general case where these assumptions are relaxed can be treated similar to \cite{pooladsanj2022saturation}. We should emphasize, however, that the merging speed assumption can affect the performance of a ramp metering policy; see Remark~\ref{remark:low-merge-speed}.  
\end{remark}
In order to have a complete description of the network, we also need to specify the demand at the microscopic level. We will do this in the next section.

\subsection{Demand Model}
For convenience in the analysis, we adopt a discrete time setting with time steps of size $\timestep$, where $\timestep$ is the minimum safe time headway at the free flow speed defined in Section~\ref{sec:vehicle-control}. Let \agent{s} arrive to \source{} $i \in \numsources$ according to an i.i.d. Bernoulli process with parameter $\Arrivalrate{i} \in [0,1]$. The arrival processes are assumed to be independent across different \source{s}. That is, at any given time step, the probability that a \agent{} arrives at the $i$-th \source{} is $\Arrivalrate{i}$ independent of everything else. We let $\Arrivalrate{} := [\Arrivalrate{i}]$. A \agent{} arriving to \source{} $i$, chooses to take \route{} $p \in \numpaths$ with probability $\routingmatrix_{ip} \in [0,1]$ independent of everything else. Naturally, for every \source{} $i$, $\sum_{p}\routingmatrix_{ip} = 1$, and $\routingmatrix_{ip} = 0$ for any $p$ that is not compatible with the network topology. We assume that the \agent{s} do not change their route choice while in the network. We stack the routing probabilities in a \emph{routing matrix} $\routingmatrix = [\routingmatrix_{ip}]$. Let $\Avgload{i} = \sum_{j}\sum_{p: i \in p}\Arrivalrate{j}\routingmatrix_{jp}$ be the \emph{average induced load} on \source{} $i$, where the notation $i \in p$ means that \source{} $i$'s node is on \route{} $p$. This is the average number of arrivals in the network that wants to cross \source{} $i$'s node in order to reach their destination. The performance of a ramp metering policy is quantified in terms of the maximum demand it can handle. We will formalize this in the next section.

\subsection{Ramp Metering and Performance Metric}\label{sec:ramp-metering}
To conveniently track vehicle locations in discrete time, we introduce the notion of \emph{slot}. A slot is associated with a particular point on the mainline at a particular time. Consider the maximum number of distinct slots that can be placed on a freeway segment, such that the distance between adjacent slots is $\timeheadway \speedlim + \standstilldist + \vehiclelength$, i.e., the minimum safety distance at the free flow speed as per (VC1) plus the vehicle length. Consider a configuration of these slots at $t=0$. At the end of each time step $\timestep$, each slot replaces the next slot with the last slot replacing the first slot of that segment. Without loss of generality, we let the location of the first and last slots of each segment coincide with the tail and head of that segment at the end of each time step. Thus, if the ego \agent{} is released from an \source{} at the end of a given time step, its location coincides with a slot.
\par
For $i \in \numsources$, let $\Queuelength{i}(t)$ be the set of the \route{} choices of the \agent{s} waiting at \source{} $i$, arranged in the order of their arrival, at time $t$. We use $|\Queuelength{i}(t)|$ to denote the queue size at \source{} $i$ at time $t$. Let $|\Queuelength{}(t)| = [|\Queuelength{i}(t)|]$ be the vector of queue sizes at all the \source{s} at time $t$. A ramp metering policy is a function $\pi(t) = [\pi_i(t)]$, where $\pi_i(t)=1$ if a vehicle is released from \source{} $i$ at time $t$, and $\pi_i(t)=0$ otherwise. The key metric we use to evaluate the performance of a ramp metering policy is its throughput, which is defined as follows: for a given ramp metering policy $\pi$ and routing matrix $\routingmatrix$, let  $U_{\pi, \routingmatrix}$ be the set of $\Arrivalrate{}$'s for which $\limsup_{t \to \infty} \E{|\Queuelength{i}(t)|} < \infty$ for all $i \in \numsources$. We say that the network is under-saturated if $\Arrivalrate{} \in U_{\pi, \routingmatrix}$; otherwise, it is called saturated. The throughput is the boundary of the set $U_{\pi, \routingmatrix}$. We are interested in finding a ramp metering policy $\pi$ such that for every $\routingmatrix$ and any other policy $\pi'$, we have $U_{\pi',\routingmatrix} \subseteq U_{\pi,\routingmatrix}$. We say such a policy maximizes the throughput.



\section{Main Results}\label{sec:results}
\subsection{Dynamic Release Rate Policy with Rate Allocation}\label{sec:release-rate}
\begin{figure}[ht]
    \centering
    \includegraphics[width=0.38\textwidth]{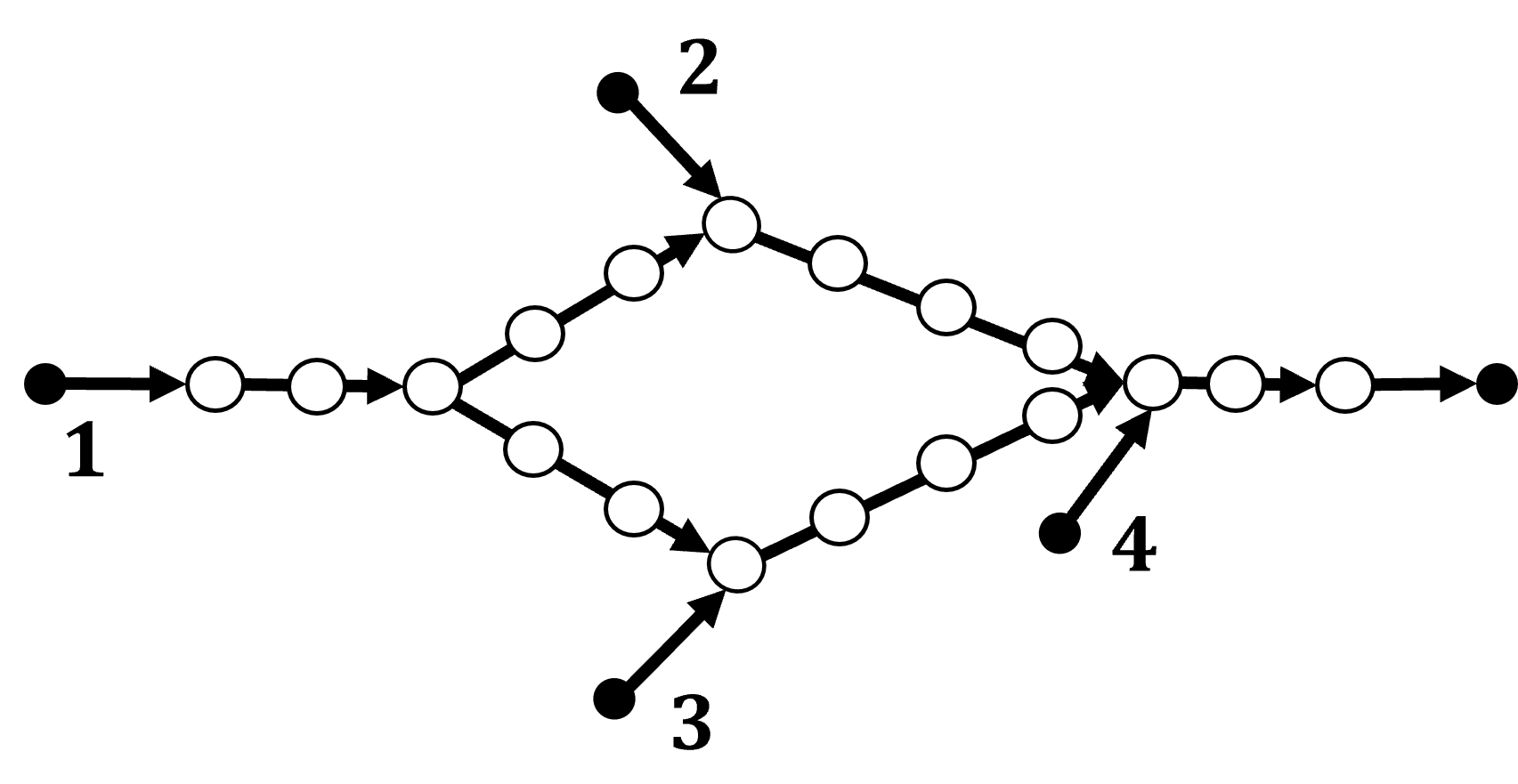}
    
    \vspace{0.1 cm}
    
    \caption{\sf A configuration of slots (empty circles) at time $t=0$. For this configuration and by choosing $b_i=3$, $a_i=1$, $i \in [3]$, and $b_4=a_4=1$ the release times can be made conflict-free.}
    \label{fig:network-example-slot}
\end{figure}
In this section we present and analyze our ramp metering policy. The policy presented is coordinated, in the sense that each \source{} may require information from other parts of the network in addition to the information obtained from its vicinity. However, it does not require any information about the demand or the route choice of the \agent{s}, i.e., it is \emph{reactive} \cite{papamichail2008traffic}. The policy operates under synchronous \emph{cycles} of fixed length during which each \source{} does not release more \agent{s} than its queue size at the beginning of the cycle. The synchronization of cycles is done once offline. This policy allocates a fixed maximum release rate to each \source{}. In addition, it imposes dynamic minimum time gap criterion between release of successive \agent{s} from the same \source{}. Changing the time gap between release of successive \agent{s} by an \source{} is akin to changing its release rate. An inner approximation to the throughput of this policy is provided, and is then compared to an outer approximation.
\begin{definition}\label{Def: DRR-RM Policy}
\textbf{(Dynamic Release Rate policy with rate Allocation ($\DRR$))} 
The policy works in cycles of fixed length $\Cyclelength \timestep$, where $\Cyclelength \in \N$. At the beginning of the $k$-th cycle at $t_k = (k-1)\Cyclelength \timestep$, each \source{} allocates itself a ``quota" equal to the queue size at that \source{} at $t_k$. At time $t \in [t_k, t_{k+1}]$ during the $k$-th cycle, \source{} $i$ releases the ego \agent{} if:
\begin{itemize}
\item[(M1)] $t = (b_{i}m + n)\timestep$, where $m \in \N_{0}$, $n \in \{n_{1}, \ldots, n_{a_{i}}\} \subseteq [b_i]$, and $a_i, b_i \in \N$, $a_i \leq b_i$, are such that the release times at all the \source{s} are \emph{conflict-free}, meaning that no two \agent{s} that move at the constant speed $\speedlim$ after being released can simultaneously reach a merging node; see Figure~\ref{fig:network-example-slot}. The term $a_i/b_i$ is the maximum release rate allocated to \source{} $i$. 
\item [(M2)] the \source{} has not reached its quota. 
\item[(M3)] the ego \agent{} is at a safe distance with respect to its leading and following \agent{s} at the moment of release.
\item[(M4)] at least $\Releasetime{}(t)$ time has passed since the release of the last \agent{} from \source{} $i$.
\end{itemize}
\par
Once an \source{} reaches its quota, it does not release a \agent{} during the rest of the cycle. The minimum time gap $\Releasetime{}(t)$ is piecewise constant, updated periodically at $t = \updateperiod, 2\updateperiod, \ldots$, as described in Algorithm~\ref{alg:DRR-rate-aloc}. In Algorithm \ref{alg:DRR-rate-aloc}, $\StateofDySys_f$ is defined as follows: for the ego \agent{}, let $x_e^T=((y_e - S_e)\indic{y_e < S_e}, v_e - \speedlim, a_e)$, where $\indic{y_e < S_e}$ is the indicator function, and $a_e$ is the acceleration. The state of all the \agent{s} is collectively denoted as $\StateofDySys^T=\left(x_e^T\right)_{e \in [n]}$, where $n$ is the total number of \agent{s} on the mainline. We let $\StateofDySys_{f_1}(t) := \|\StateofDySys(t)\|$. Moreover, for $t \geq \updateperiod$, let $\StateofDySys_{f_2}(t) := \sum |\hat{y}_e(t_m) - \hat{S}_e(t_m)|$, where the sum is over all the \agent{s} in the time interval $[t-\updateperiod,t]$ that predicted to violate the safety distance after reaching a merging node. We let $\StateofDySys_{f}(t) = \StateofDySys_{f_1}(t) + \StateofDySys_{f_2}(t)$. Roughly, $\StateofDySys_{f}(t)$ measures the distance to the free flow equilibrium state at time $t$.
\begin{algorithm}[H]
\caption{Update rule for the minimum time gap between release of successive \agent{s}}\label{alg:DRR-rate-aloc}
\begin{algorithmic}
\Require \textbf{design constants:}~$\updateperiod>0, \DesDecreaseconst > 0$, $\ReleasetimeDecConst{2} > 0$, $\ReleasetimeInc{}^{\circ} > 0$,
\par
\hspace{2.6cm} $\ReleasetimeAdj{} > 1$
\par
~\textbf{initial condition:}~$\Releasetime{}(0) = 0,~ \ReleasetimeInc{}(0) = \ReleasetimeInc{}^{\circ}$,
\par
\hspace{2.6cm} $\StateofDySys_f(0) = \StateofDySys_{f_1}(0)$

\hspace{-0.475in} \textbf{for} $t = \updateperiod, 2\updateperiod, \cdots$ \textbf{do} 

\If{$\StateofDySys_f(t) \leq \max\{\StateofDySys_f(t-\updateperiod) - \DesDecreaseconst{}, 0\}$ } 
\State $\ReleasetimeInc{}(t) \gets \ReleasetimeInc{}(t-\updateperiod)$
\State $\Releasetime{}(t) \gets \max\{\Releasetime{}(t-\updateperiod) - \ReleasetimeDecConst{2}, 0\}$
\Else{}
\State $\ReleasetimeInc{}(t) \gets \ReleasetimeAdj{} \ReleasetimeInc{}(t-\updateperiod)$
\State $\Releasetime{}(t) \gets \Releasetime{}(t-\updateperiod) + \ReleasetimeInc{}(t)$

\EndIf

\hspace{-0.475in} \textbf{end for}
\end{algorithmic}
\end{algorithm}
\end{definition}
\begin{theorem}\label{Prop: stability of DRR-RM}
For any initial condition, $\Cyclelength \in \N$, design constants in Algorithm~\ref{alg:DRR-rate-aloc}, and $a_i, b_i \in \N$ such that $a_i \leq b_i$, $i \in \numsources$, and the release times are conflict-free, the $\DRR$ policy keeps the network under-saturated if 
   $ \Avgload{i} < a_i/b_i$ for all $i \in \numsources$.
\end{theorem} 
\begin{proof}
See Appendix \ref{Section: (Appx) Proof of stability of DRR-RM}.
\comment{
A sketch is provided in Appendix~\ref{sec:proof-sketch}. The complete proof can be found in the arXiv version of this paper.
}
\end{proof}
\textbf{V2I communication requirements}: This policy uses information about $|\Queuelength{}|$ (if $\Cyclelength > 1$) and $\StateofDySys_f$. After a finite time, the mainline \agent{s} only communicate $\StateofDySys_{f_1}$ to all the \source{s} at the end of each update period $\updateperiod$. The number of \agent{s} that constitute $\StateofDySys_{f_1}$ is no more than the total number of slots. Furthermore, at the end of each time step during a cycle, the \agent{s} near an \source{} communicate their state to that \source{} for the safety gap evaluation. 
\begin{remark}\label{remark:nonreactive-DRRA}
\textbf{A non-reactive version of the $\DRR$ policy with better throughput}: suppose that the first \agent{} in every \source{}'s queue also communicates its route choice to that \source{}. Consider a  version of the $\DRR$ policy, where, in addition to the releasing times in (M1), \source{} $i$ releases the ego \agent{} at $t=m\timestep$, $m \in \N_{0}$, if the ego \agent{}'s \route{} contains no merge node. This policy is non-reactive since it uses the route choice of the \agent{s}, but gives a better throughput since more \agent{s} are released compared to the $\DRR$ policy. We compare the performance of these two policies via simulations in Section~\ref{sec:simulation}.
\end{remark}
\begin{remark}\label{remark:low-merge-speed}
Recall from Section~\ref{sec:vehicle-control} that we made the simplifying assumption that the merging speed at all the \source{s} is $\speedlim$. One can relax this assumption, but the throughput would suffer due to the larger safety distance required at the moment of merging. However, the main steps of the proof of Theorem~\ref{Prop: stability of DRR-RM} remain the same and a similar bound for the throughput can be obtained; see \cite[Theorem 2]{pooladsanj2022saturation}. 
\end{remark}

\subsection{A Necessary Condition}
\label{sec:necessary}
We now provide a necessary condition for under-saturation, against which we benchmark the sufficient condition in Theorem~\ref{Prop: stability of DRR-RM}. Let $\Cumuldepartures{\pi,p}(m\timestep)$ be the cumulative number of vehicles that has crossed point $p$ on the mainline up to time $m\timestep$, $m \in \N_{0}$, under the ramp-metering policy $\pi$. Then, the crossing rate at point $p$ is defined as $\Cumuldepartures{\pi,p}(m\timestep)/m$ and the ``long-run" crossing rate is $\limsup_{m \ra \infty}\Cumuldepartures{\pi,p}(m\timestep)/m$. 
Recall the definition of $\Avgload{i}$ for \source{} $i$'s node. We can extend this definition to every node in the network that is not a source or sink. Let $\Avgload{} := \max_{i}\Avgload{i}$, where $i$ is any node that is not a source or sink.
\begin{theorem}\label{Prop: Necessary condition for stability}
Suppose that the long-run crossing rate is no more than one vehicle per $\timestep$ seconds for at least one point on each freeway segment. If the network is under-saturated under the policy $\pi$, then $\Avgload{} \leq 1$.
\end{theorem}
\begin{proof}
The proof is similar to \cite[Theorem 4]{pooladsanj2022saturation}.
\end{proof}
\begin{remark}
The long-run crossing rate condition in Theorem~\ref{Prop: Necessary condition for stability} is very mild in practice. In particular, if the traffic flow at some point on each freeway segment is less than the mainline capacity, then the long-run crossing rate condition is also satisfied \cite{pooladsanj2022saturation}.
\end{remark}
\begin{remark}
Consider a network with no merge node, and note that for any node $j$ that is not a source or sink, we have $\Avgload{j} \leq \Avgload{i}$ for some $i \in \numsources$. For such a network, one can set $a_i=b_i=1$ for all $i \in \numsources$ in the $\DRR$ policy. The sufficient condition of Theorem~\ref{Prop: stability of DRR-RM} then becomes $\max_{i \in \numsources} \Avgload{i} = \Avgload{} < 1$. Combining with Theorem~\ref{Prop: Necessary condition for stability}, it follows for any routing matrix $\routingmatrix$ and any other ramp metering policy $\pi'$ that $U_{\pi',\routingmatrix} \subseteq U_{\DRR, \routingmatrix}$, except, maybe at the boundary of $U_{\DRR, \routingmatrix}$. Since the boundary of $U_{\DRR, \routingmatrix}$ has zero volume (measure), we can conclude that the $\DRR$ policy maximizes the throughput for all practical purposes.
\end{remark}

\section{Simulations}\label{sec:simulation}
\begin{figure}[t!]
        \centering
        \includegraphics[width=0.45\textwidth]{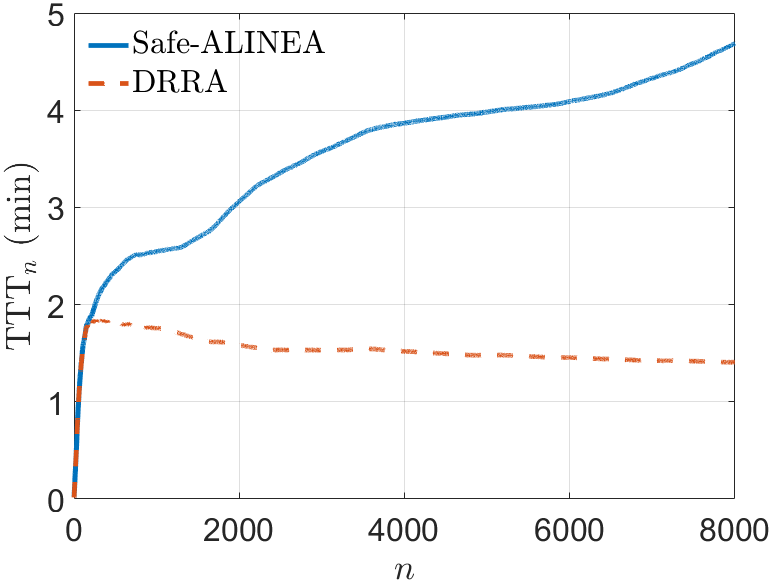}
        \vspace{0.2 cm}
    \caption{\centering\sf The average total travel time ($\text{TTT}_n$) under the Safe-ALINEA and $\DRR$ policies, where $n$ is the number of completed trips.}\label{Fig: Comaring CAD}
\end{figure}
\begin{figure}[htb!]
\begin{center}
    \begin{subfigure}{.35\textwidth}
        \centering
        \includegraphics[width=\textwidth]{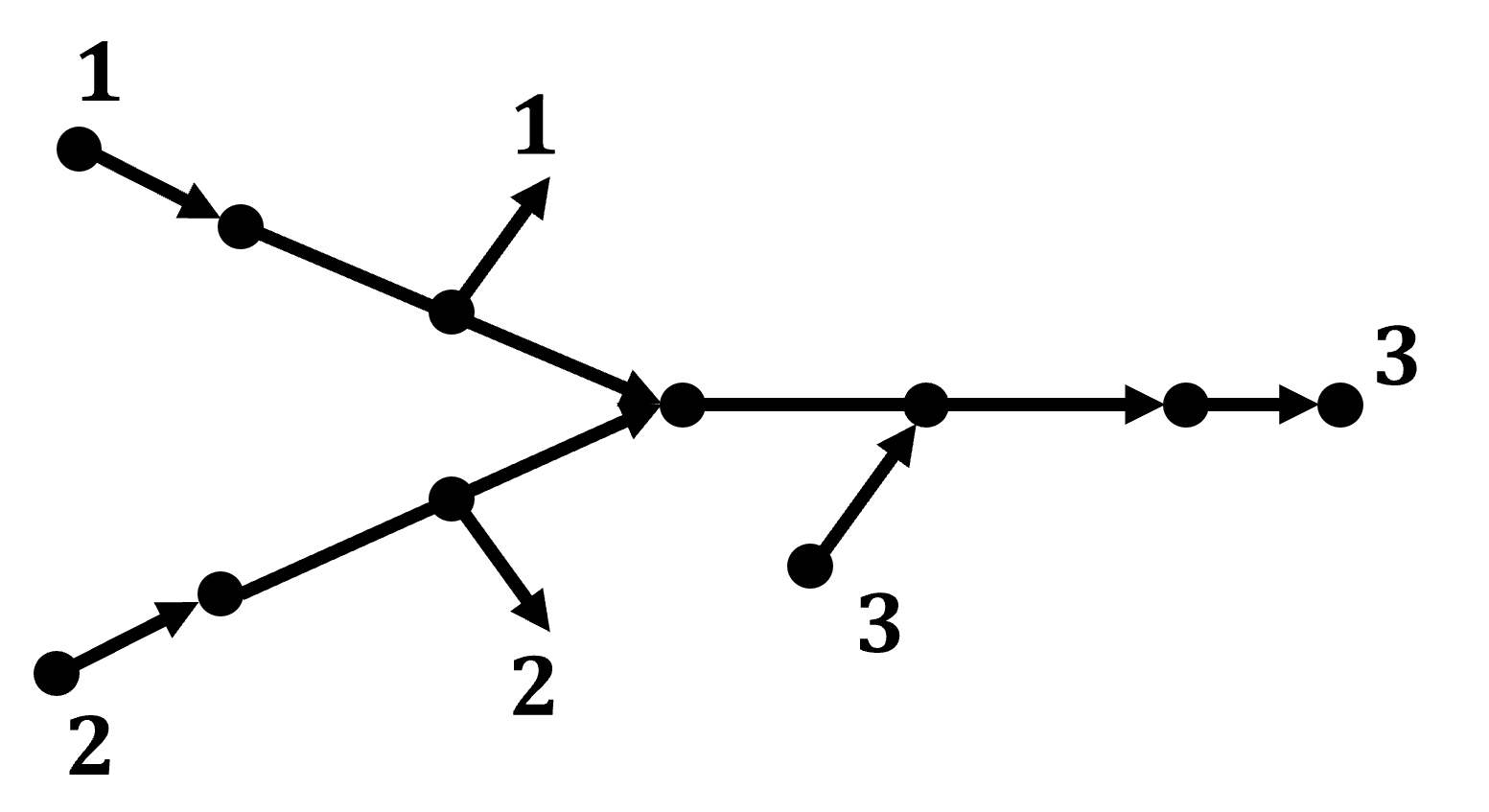}
        \caption{} \label{fig:first-network-sim}
    \end{subfigure}
    \begin{subfigure}{.35\textwidth}
    \vspace{0.2 cm}
        \centering
        \includegraphics[width=\textwidth]{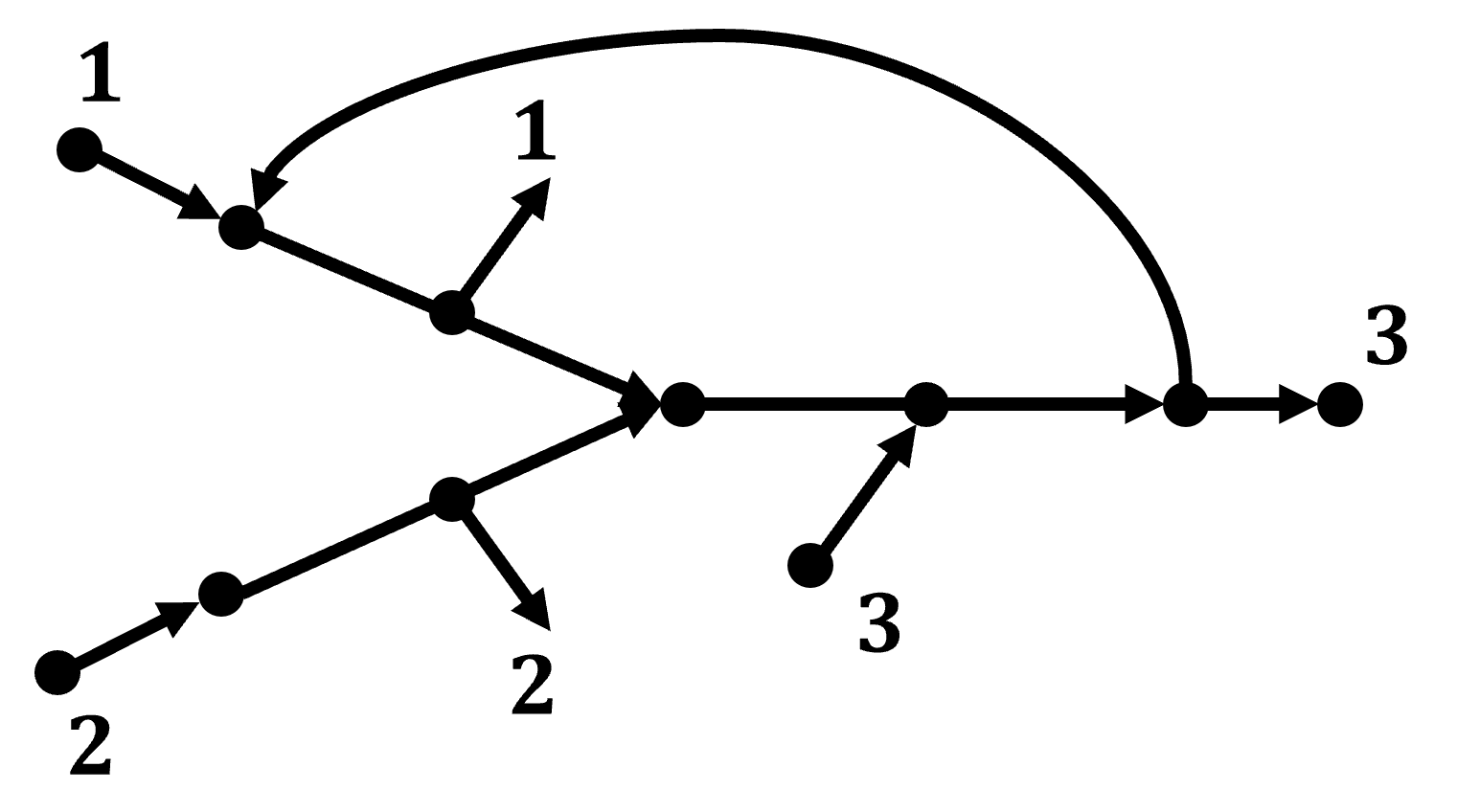}
        \caption{} \label{fig:second-network-sim}
    \end{subfigure}
    \end{center}
    \vspace{0.1 cm}
    \caption{\sf The two networks used in Sections~\ref{sec:network-1} and \ref{sec:network-2}. Figure (a) is a three-legged merge junction, and figure (b) is constructed by making the first network cyclic.}
\end{figure}
We simulate the $\DRR$ policy and compare its performance with its non-reactive version introduced in Remark~\ref{remark:nonreactive-DRRA} and a well-known policy from the literature. All the simulations were performed in MATLAB. We use the following setup in all the simulations: let $\timeheadway = 1.5 ~[\text{s}]$, $\standstilldist = 4~ [\text{m}]$, $\vehiclelength = 4.5 ~[\text{m}]$, and $\speedlim = 15 ~[\text{m/s}]$. The initial queue size at all the \source{s} is set to zero, and \agent{s} arrive at the \source{s} according to i.i.d Bernoulli processes with the same rate $\Arrivalrate{}$.
\subsection{Total Travel Time}\label{sec:travel-time}
In this section, we compare the total travel time under the $\DRR$ policy and ALINEA \cite{papageorgiou1991alinea}, which is a macroscopic ramp metering policy. We also add the safety filter (M3) on top of ALINEA, and use the name Safe-ALINEA to refer to the resulting policy. An \source{}'s outflow in the ALINEA policy is determined according to the following equation:
\begin{equation*}
    r(k) = r(k-1) + K_r(\hat{o}-o(k)).
\end{equation*}
\par
Here, $r(\cdot)$ is the \source{}'s outflow, $K_r$ is a positive design constant, $o(\cdot)$ is the occupancy of the mainline downstream of the \source{}, and $\hat{o}$ is the desired occupancy. We use time steps of size $60~[\text{s}]$, $K_r = 70~[\text{veh/h}]$ \cite{papageorgiou1991alinea}, and $\hat{o} = 13~\%$ corresponding to the capacity flow. For the $\DRR$ policy, we use the following parameters: $\Cyclelength = 13$, $\updateperiod=2\timestep$, $\DesDecreaseconst = 50$, $\ReleasetimeDecConst{2} = 10$, $\ReleasetimeInc{}^{\circ} = 0.1$, $\ReleasetimeAdj{} = 1.01$, and $a_i=b_i=1$, $i=1,2,3$. The network is a ring road with perimeter $1860~[\text{m}]$ and $3$ on- and \sink{s}. The \source{s} are located symmetrically at locations $0$, $600$, and $1200$; the \sink{s} are also located symmetrically $155~[\text{m}]$ upstream of each \source{}. We consider a congested initial condition where the initial number of \agent{s} on the mainline is $100$; each \agent{} is at the constant speed of $6.7~[\text{m/s}]$, and is at the distance $\timeheadway \times 6.7 + \standstilldist \approx 14.1~[\text{m}]$ from its leading \agent{}. The details of the \agent{} following behavior is given in \cite{pooladsanj2022saturation}. We also let $\Arrivalrate{} = 0.54$ and 
\begin{equation*}
    \routingmatrix_1 = \begin{pmatrix}
    0.2 & 0.7 & 0.1 \\
    0 & 0.8 & 0.2 \\
    0.5 & 0 & 0.5
    \end{pmatrix}, 
\end{equation*}
where columns specify the \sink{s}.
\par
We evaluate the average Total Travel Time (TTT$_n$), which is computed by averaging the total travel time of the first $n$ vehicles that complete their trips. The total travel time of a vehicle equals its waiting time in an on-ramp queue plus the time it spends on the freeway to reach its destination. Figure~\ref{Fig: Comaring CAD} shows the simulation result. From the figure, we can see that the $\DRR$ policy provides significant improvement in the TTT$_n$ compared to the Safe-ALINEA policy. In fact, the network becomes saturated under the Safe-ALINEA policy, which implies that the TTT$_n$ will grow steadily with $n$.
\subsection{Acyclic Network with Merge Junction}\label{sec:network-1}
The second network is a symmetric three-legged merge junction; see Figure~\ref{fig:first-network-sim}. The total length of freeway segments on each leg is $310~[\text{m}]$. \source{s} $1$ and $2$ are located at the beginning, while \source{} $3$ is located in the middle of their corresponding legs. Similarly, \sink{s} $1$ and $2$ are in the middle, while \sink{} $3$ is at the end of their corresponding legs. The mainline is initially empty and 
\begin{equation*}
    \routingmatrix_2 = \begin{pmatrix}
    0.6 & 0 & 0.4 \\
    0 & 0.6 & 0.4 \\
    0 & 0 & 1
    \end{pmatrix}.
\end{equation*}
\par
The \source{s} use the $\DRR$ policy with $\Cyclelength = 1$ and other parameters chosen similar to Section~\ref{sec:travel-time}. In addition, the release times of \source{s} $1$ and $2$ are set to $t=(2m+1)\timestep$ and $t=(2m+2)\timestep$, respectively, while $t=m\timestep$ for \source{} $3$, $m \in \N_{0}$ (thus, $a_i=1$, $b_i=2$, $i \in [2]$, and $a_3=b_3=1$). For the given routing matrix, the inner and outer estimates of the under-saturation region given by Theorems~\ref{Prop: stability of DRR-RM} and \ref{Prop: Necessary condition for stability} are $\Arrivalrate{} < 1/2$ and $\Arrivalrate{} < 5/9$, respectively. From simulations, the under-saturation regions of the $\DRR$ policy and its non-reactive version are estimated to be $\Arrivalrate{} < 1/2$ and $\Arrivalrate{} < 5/9$, respectively. This suggest that there is a gap between the inner and outer estimates of the under-saturation region in the $\DRR$ policy, but not in its non-reactive version. 
\subsection{Cyclic Network with Merge Juction}\label{sec:network-2}
The third network is constructed by making the second network cyclic; see Figure~\ref{fig:second-network-sim}. The length of the additional segment in this network is $610~[\text{m}]$. The mainline is initially empty and
\begin{equation*}
    \routingmatrix_3 = \begin{pmatrix}
    0.6 & 0 & 0.4 \\
    0 & 0.6 & 0.4 \\
    0.5 & 0 & 0.5
    \end{pmatrix}.
\end{equation*}
\par
The \source{s} again use the $\DRR$ policy with the same parameters as in the previous section. The inner and outer estimates of the under-saturation region are $\Arrivalrate{} < 1/3$ and $\Arrivalrate{} < 5/9$, respectively. From simulations, the under-saturation regions of the $\DRR$ policy and its non-reactive version are estimated to be $\Arrivalrate{} < 0.4$ and $\Arrivalrate{} < 5/9$, respectively. Therefore, Theorem~\ref{Prop: stability of DRR-RM} holds for this cyclic network. Based on the simulation results and the results in \cite{pooladsanj2022saturation} given for a ring road network, we conjecture that Theorem~\ref{Prop: stability of DRR-RM} holds for general cyclic networks. 

\section{Conclusion and Future Work}\label{sec:conclusion}
In this paper, we provided a microscopic ramp metering policy for freeway networks with arbitrary number of on- and off-ramps, merge, and diverge junctions. We analyzed the throughput of our policy for acyclic networks, and showed via simulations that it achieves a similar throughput for cyclic networks. We plan to confirm this observation analytically in the future. Moreover, we plan to extend our analysis to policies with tighter estimates of the throughput.  

\bibliographystyle{ieeetr}
\bibliography{References_main}

\begin{thebibliography}{10}

\bibitem{pooladsanj2022saturation}
M.~Pooladsanj, K.~Savla, and P.~A. Ioannou, ``Saturation region of freeway
  networks under safe microscopic ramp metering,'' {\em arXiv preprint
  arXiv:2207.02360}, 2022.

\bibitem{papageorgiou2002freeway}
M.~Papageorgiou and A.~Kotsialos, ``Freeway ramp metering: An overview,'' {\em
  IEEE transactions on intelligent transportation systems}, vol.~3, no.~4,
  pp.~271--281, 2002.

\bibitem{papageorgiou2003review}
M.~Papageorgiou, C.~Diakaki, V.~Dinopoulou, A.~Kotsialos, and Y.~Wang, ``Review
  of road traffic control strategies,'' {\em Proceedings of the IEEE}, vol.~91,
  no.~12, pp.~2043--2067, 2003.

\bibitem{wattleworth1965peak}
J.~A. Wattleworth, ``Peak-period analysis and control of a freeway system,''
  tech. rep., Texas Transportation Institute, 1965.

\bibitem{papageorgiou1991alinea}
M.~Papageorgiou, H.~Hadj-Salem, J.-M. Blosseville, {\em et~al.}, ``Alinea: A
  local feedback control law for on-ramp metering,''

\bibitem{papageorgiou1990modelling}
M.~Papageorgiou, J.-M. Blosseville, and H.~Haj-Salem, ``Modelling and real-time
  control of traffic flow on the southern part of boulevard
  p{\'e}riph{\'e}rique in paris: Part ii: Coordinated on-ramp metering,'' {\em
  Transportation Research Part A: General}, vol.~24, no.~5, pp.~361--370, 1990.

\bibitem{papamichail2010coordinated}
I.~Papamichail, A.~Kotsialos, I.~Margonis, and M.~Papageorgiou, ``Coordinated
  ramp metering for freeway networks--a model-predictive hierarchical control
  approach,'' {\em Transportation Research Part C: Emerging Technologies},
  vol.~18, no.~3, pp.~311--331, 2010.

\bibitem{gomes2006optimal}
G.~Gomes and R.~Horowitz, ``Optimal freeway ramp metering using the asymmetric
  cell transmission model,'' {\em Transportation Research Part C: Emerging
  Technologies}, vol.~14, no.~4, pp.~244--262, 2006.

\bibitem{stern2018dissipation}
R.~E. Stern, S.~Cui, M.~L. Delle~Monache, R.~Bhadani, M.~Bunting, M.~Churchill,
  N.~Hamilton, H.~Pohlmann, F.~Wu, B.~Piccoli, {\em et~al.}, ``Dissipation of
  stop-and-go waves via control of autonomous vehicles: Field experiments,''
  {\em Transportation Research Part C: Emerging Technologies}, vol.~89,
  pp.~205--221, 2018.

\bibitem{pooladsanj2020vehicle}
M.~Pooladsanj, K.~Savla, and P.~Ioannou, ``Vehicle following over a closed ring
  road under safety constraint,'' in {\em 2020 IEEE Intelligent Vehicles
  Symposium (IV)}, pp.~413--418, IEEE.

\bibitem{pooladsanj2021vehicle}
M.~Pooladsanj, K.~Savla, and P.~A. Ioannou, ``Vehicle following on a ring road
  under safety constraints: Role of connectivity and coordination,'' {\em IEEE
  Transactions on Intelligent Vehicles}, vol.~8, no.~1, pp.~628--638, 2022.

\bibitem{rios2016automated}
J.~Rios-Torres and A.~A. Malikopoulos, ``Automated and cooperative vehicle
  merging at highway on-ramps,'' {\em IEEE Transactions on Intelligent
  Transportation Systems}, vol.~18, no.~4, pp.~780--789, 2016.

\bibitem{Malikopoulos.2017}
J.~{Rios-Torres} and A.~A. {Malikopoulos}, ``A survey on the coordination of
  connected and automated vehicles at intersections and merging at highway
  on-ramps,'' {\em IEEE Transactions on Intelligent Transportation Systems},
  vol.~18, pp.~1066--1077, May 2017.

\bibitem{Ioannou.Chien.1993}
P.~Ioannou and C.~Chien, ``Autonomous intelligent cruise control,'' {\em IEEE
  Transactions On Vehicular Technology}, vol.~42, no.~4, pp.~657--672, 1993.

\bibitem{papamichail2008traffic}
I.~Papamichail and M.~Papageorgiou, ``Traffic-responsive linked ramp-metering
  control,'' {\em IEEE Transactions on Intelligent Transportation Systems},
  vol.~9, no.~1, pp.~111--121, 2008.

\bibitem{meyn2012markov}
S.~P. Meyn and R.~L. Tweedie, {\em Markov chains and stochastic stability}.
\newblock Springer Science \& Business Media, 2012.

\bibitem{folland1999real}
G.~B. Folland, {\em Real analysis: modern techniques and their applications},
  vol.~40.
\newblock John Wiley \& Sons, 1999.

\end{thebibliography}

\appendices
\section{Performance Analysis Tool}
\label{sec:performance-analysis}
Consider an initial condition where the \agent{s} are in the free flow equilibrium, and the location of each \agent{} coincides with a slot. Thus, their location coincide with a slot in the future. Thereafter, during each time step, the following sequence of events occurs: (i) the slots on each segment replace the next slot with the last slot replacing the first slot; the numbering of slots is reset with the new first slot numbered $1$, and the rest numbered in an increasing order in the direction of the flow, (ii) \agent{s} that reach their destination \sink{} exit the network, (iii) if permitted by the ramp metering policy, a \agent{} is released into a slot at the free flow speed. Since the initial location of all \agent{s} coincide with a slot and the released times are conflict-free, the location of the released \agent{} will also coincide with a slot in the future.
\par
For the aforementioned initial condition, we can cast the network dynamics as a discrete-time Markov chain, with time steps of size $\timestep$. Let $M(n)$ be the set of the \route{s} of the occupants of all the slots: $M_{\ell}(n)=p$ if the \agent{} occupying slot $\ell$ at time $n$ has to take route $p$, and $M_{\ell}(n)=0$ if slot $\ell$ is empty at time $n$. 
Consider the following discrete-time Markov chain with the state 
    $\StateofMCExp_{\triangle}(n) := \StateofMC(n\triangle)$, $n \geq 0$,
where $\StateofMC(n) := (\Queuelength{}(n),\Nodedest{}(n))$, and $\triangle \in \N$ will be specified for the policy being analyzed. The transition probabilities are determined by the ramp metering policy being analyzed, but will not be specified explicitly for brevity. For the ramp metering policy considered in this paper, the state $\StateofMC(n) = (0, 0)$ is reachable from all other states, and $\mathbb{P}\left(\StateofMC(n+1) = (0, 0)~|~ \StateofMC(n) = (0,0)\right) > 0$. Hence, the Markov chain $\StateofMCExp_{\triangle}$ is irreducible and aperiodic.

\section{Proof of Theorem \ref{Prop: stability of DRR-RM}}\label{Section: (Appx) Proof of stability of DRR-RM}

For any initial condition, we first show that there exists some $k_0 \in \N$ such that for all $k \geq k_0$, $\StateofDySys_f(k\updateperiod) = 0$. This and (VC1) imply that every \agent{} moves at the free flow speed for all $t \geq t_0$. Hence, after a finite time after $t_0$, $\Releasetime{}(\cdot)=0$, and the location of every \agent{} coincides with a slot thereafter. We can then use the Markov chain setting from Appendix~\ref{sec:performance-analysis}.
\par
Using an argument similar to \cite{pooladsanj2022saturation}, one can easily show that $X_f(k\updateperiod)$ is uniformly bounded for all $k \in \N$. To show $X_f(\cdot)=0$ after a finite time, we use a proof by contradiction. Suppose that $X_f(k\updateperiod) \neq 0$ for infinitely many $k \in \N$. Since $\StateofDySys_f(k\updateperiod)$ is uniformly bounded, there exists an infinite sequence $\{k_n\}_{n \geq 1}$ such that $\StateofDySys_f(k_n\updateperiod) > \StateofDySys_f((k_n - 1)\updateperiod) - \DesDecreaseconst$ for all $n \geq 1$. This implies that $\ReleasetimeInc{}(k_n \updateperiod) = \ReleasetimeAdj{} \ReleasetimeInc{}((k_n-1) \updateperiod)$ for all $n \geq 1$. Since $\ReleasetimeInc{}(\cdot)$ is non-decreasing and $\ReleasetimeAdj{} > 1$,  it follows that $\lim_{t \ra \infty}\ReleasetimeInc{}(t) = \infty$, which in turn implies that $\limsup_{t \ra \infty}\Releasetime{}(t) = \infty$. Let $k_f$ be such that $\Releasetime{}(k_f\updateperiod) > |\numsources| (\Tempty + T_{\text{empty}})(1 + \ReleasetimeDecConst{2}/\updateperiod)$, where $T_{\text{empty}}$ is an upper bound on the additional time required for all \agent{s} on the mainline to leave the network after reaching the free flow state, if no additional \agent{} is released. Let $t_f := k_f\updateperiod$. Note that for all $t \in [t_f, t_f + |\numsources|(\Tempty+T_{\text{empty}})]$, $g(t) > |\numsources|(\Tempty+T_{\text{empty}})$. Thus, each \source{} releases at most one \agent{} during the interval $[t_f, t_f + |\numsources|(\Tempty+T_{\text{empty}})]$. Hence, there exists a time interval of length at least $(\Tempty+T_{\text{empty}})$ in $[t_f, t_f + |\numsources|(\Tempty+T_{\text{empty}})]$ during which no \source{} releases a \agent{}. Condition (VC2) then implies that all \agent{s} reach the free flow state after at most $\Tempty$ time, and leave the network after at most an additional $T_{\text{empty}}$ time, at the end of which (say at time $t_0$), the network becomes empty. Thus, for all $k \geq k_0 := \lceil t_0/\updateperiod \rceil$, $\StateofDySys_f(k\updateperiod) = 0$, which is a contradiction to the assumption that $\StateofDySys_f(k\updateperiod)\neq 0$ for infinitely many $k$.
\par
\begin{figure}[t]
    \centering
    \includegraphics[width=0.38\textwidth]{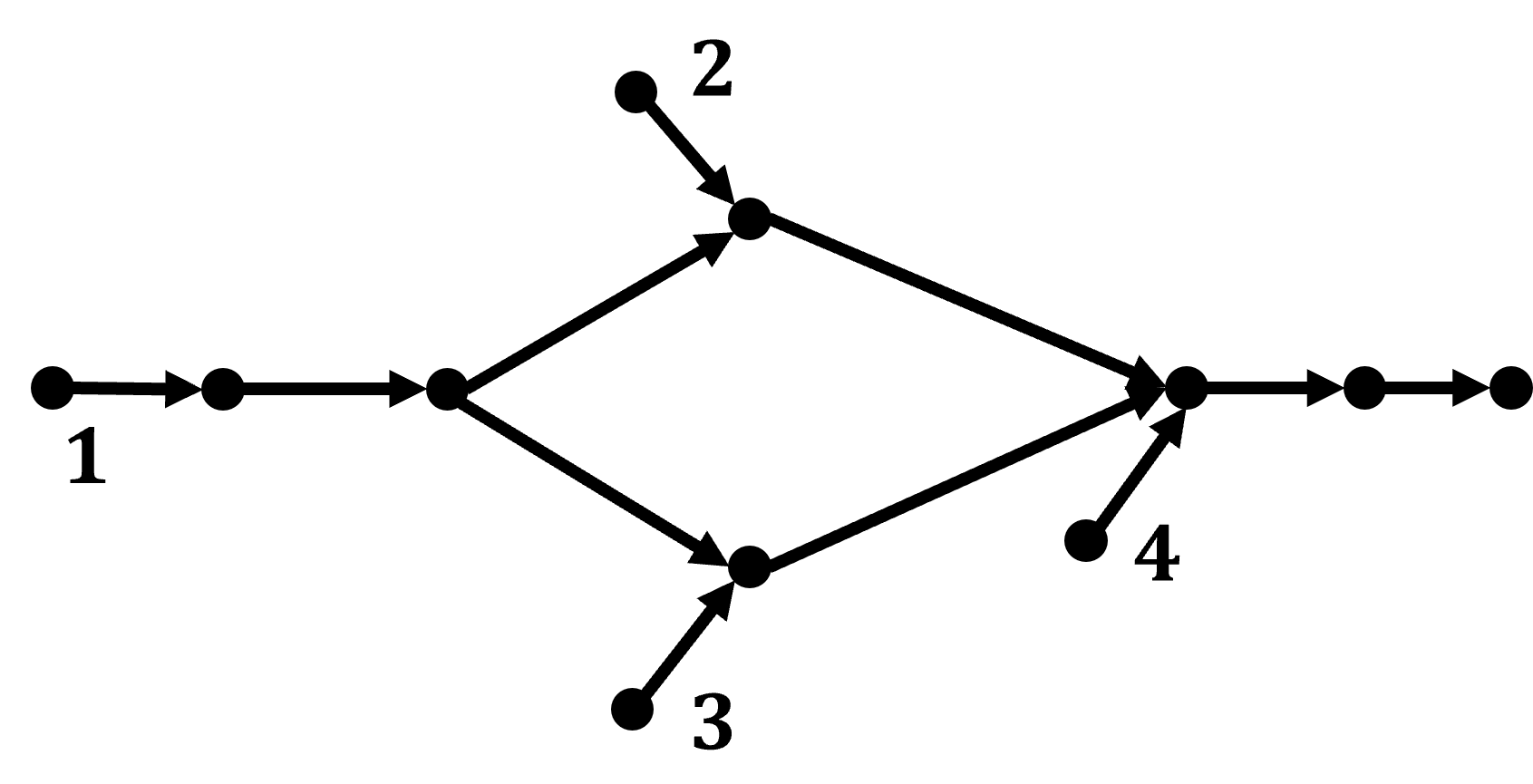}
    
    \vspace{0.1 cm}
    
    \caption{\sf For \source{} $4$, $U_{4}^{0} = \{\{4\}\}$, $U_{4}^{1} = \{\{2,3\}\}$, and $U_{4}^{2} = \{\{1,1\}, \{1,3\}, \{1,2\}\}$.}
    \label{fig:network-example-proof}
\end{figure}
We assume that the \agent{s} are initially in the free flow equilibrium state, and the location of each \agent{} coincides with a slot at $t=0$ as in Appendix~\ref{sec:performance-analysis}. For the sake of readability, we present proofs of intermediate claims at the end. We adopt the Markov chain setting from Appendix~\ref{sec:performance-analysis} with $\{\StateofMCExp_{\triangle}(n)\}_{n \geq 0}$ as the Markov chain, where $\triangle = k\Cyclelength$ and $k \in \N$ is to be determined. For \source{} $i \in \numsources$, let $\Nodedegree{i}(n)$ be the \emph{degree} of \source{} $i$ representing the number of \agent{s} in the network at time $n$ that need to cross \source{} $i$'s node in order to reach their destination. We recursively define the following collections of index sets:
\begin{equation*}
        U_{i}^{0} = \{\{i\}\},~ U_{i}^{1} = \{\numsources_{i}^{-}\},  
\end{equation*}
and for $k \geq 2$,
\begin{equation*}
    U_{i}^{k} = \{\cup_{j \in I} I_{j}:~ I \in U_{i}^{k-1}, I_j = \{j\}~\text{or}~\numsources_{j}^{-}\}\setminus U_{i}^{k-1},
\end{equation*}
where each $\cup_{j \in I} I_{j}$ is treated as a \emph{multiset}; see Figure~\ref{fig:network-example-proof}. In words, $I \in U_{i}^{1}$, i.e., $I = \numsources_{i}^{-}$, is the set of all \source{s} that precede \source{} $i$; if $I \in U_{i}^{2} \cup U_{i}^{1}$, then for any \source{} $j$ that precedes \source{} $i$, either $j \in I$ or all \source{s} that precede $j$ belong to $I$; and so on. Note that since the network is acyclic, the number of sets $U_{i}^{k}$ is finite for every $i$. Let $U = \cup_{i \in \numsources, k \geq 0}U_{i}^{k}$. Consider the candidate Lyapunov function $V: \StateSpace \ra [0,\infty)$ defined as: 
\begin{equation}
\label{eq:lyap-def}
 V(n) \equiv \Lyap{\StateofMCExp_{\triangle}(n)} := \Nodedegree{}^2(n\triangle),
\end{equation}
where $\StateSpace$ is the range of values of $Y$ in Appendix~\ref{sec:performance-analysis}, and
    $\Nodedegree{}(m) = \max\limits_{I \in U} \sum_{j \in I} \Nodedegree{j}(m)$.
\par
Note that \eqref{eq:lyap-def} implies $V(n+1) - V(n) = \Nodedegree{}^2((n+1)\triangle) - \Nodedegree{}^2(n\triangle)$. We show at the end of the proof that if $\Lyap{n}$ is large enough, specifically if $\Lyap{n} > L := \left((C_1+1)\triangle^2 + C_2\right)^2$, where $C_1 := (|\numsources|+\Avgload{}+\delta)|\numsources||\numpaths|$ and $C_2:= (2\Cyclelength +\max_{j \in \numsources}( \numcells{j}+a_{j}))|\numsources||\numpaths|$, then
\begin{equation}\label{eq:node-degree-inequality}
    \begin{aligned}
    \Nodedegree{}((n+1)\triangle) \leq \Nodedegree{}(n\triangle) + \MaxNumArrival{}{}{\triangle} + C_2,
    \end{aligned}
\end{equation}
where $\MaxNumArrival{}{}{\triangle}$ satisfies the following: it is independent of $\Nodedegree{}(n\triangle)$, and there exists $\triangle_1, \eps, C_3 > 0$ such that for all $\triangle > \triangle_1$ we have
\begin{equation}\label{eq:Kolmogorov-implication}
    \E{\MaxNumArrival{}{}{\triangle}} < -\eps \triangle, \quad \E{\MaxNumArrival{}{2}{\triangle}} < C_3 \triangle^2.
\end{equation}
The inequality \eqref{eq:node-degree-inequality} implies that
\begin{equation}\label{eq:node-degree-inequality-squared}
\begin{aligned}
        \Nodedegree{}^2((n+1)\triangle) \leq \Nodedegree{}^2(n\triangle) &+ 2\Nodedegree{}(n\triangle)\left(\MaxNumArrival{}{}{\triangle}+C_2\right)\\
        &+ \MaxNumArrival{}{2}{\triangle} + 2C_2\MaxNumArrival{}{}{\triangle} + C_{2}^{2}.
\end{aligned}
\end{equation}
\par
By choosing $\triangle \geq \triangle_1$ and taking conditional expectation from both sides of \eqref{eq:node-degree-inequality-squared} we obtain 
\begin{equation*}
\label{eq:lyap-diff}
\begin{aligned}
        \mathbb{E}\left[\right.&\left.\Lyap{n+1} - \Lyap{n}\:|\: V(n) > L\right] \leq -\Nodedegree{}(n\triangle) \\
        &+ 2\Nodedegree{}(n\triangle)(-\eps \triangle + C_2 + \frac{1}{2}) + C_3\triangle^2 -  2\eps C_2\triangle + C_{2}^{2}.
\end{aligned}
\end{equation*}
\par
Since for any $i \in \numsources$, $|\Queuelength{i}(n\triangle)| \leq \Nodedegree{i}(n\triangle) \leq \Nodedegree{}(n\triangle)$, we have $\|\Queuelength{}(n\triangle)\|_{\infty} := \max_{i \in \numsources}|\Queuelength{i}(n\triangle)| \leq \Nodedegree{}(n\triangle)$. Hence, choosing $\triangle > \max\{\triangle_1,\triangle_2 := \frac{C_2+1/2}{\eps}\}$ gives
\begin{equation*}
\begin{aligned}
        \mathbb{E}\left[\right.&\left.\Lyap{n+1} - \Lyap{n}\:|\: V(n) > L\right] \leq -\|\Queuelength{}(n\triangle)\|_{\infty} \\ &-2\eps(C_1+1)\triangle^3
        + \left((C_1+1)(2C_2+1)+C_3\right)\triangle^2 \\
        &- 4\eps C_2\triangle + 3C_2^2+1.
\end{aligned}
\end{equation*}
\par
If $\triangle_3$ is chosen such that for all $\triangle > \max\{\triangle_1,\triangle_2,\triangle_3\}$, 
\begin{equation}\label{eq:sufficient-large-triangle}
\begin{aligned}
    -2\eps(C_1+1)\triangle^3 &+ \left((C_1+1)(2C_2+1)+C_3\right)\triangle^2 \\
    &-4\eps C_2\triangle + 3C_2^2+1 < 0,
\end{aligned}
\end{equation}
then $\E{\Lyap{n+1} - \Lyap{n}|\: \Lyap{n} > L} \leq -\|\Queuelength{}(n\triangle)\|_{\infty}$. Such a $\triangle_3$ always exists because the $-2\eps(C_1+1)\triangle^3$ term in \eqref{eq:sufficient-large-triangle} dominates for sufficiently large $\triangle$. We let $\tilde{\triangle} > \max\{\triangle_1,\triangle_2, \triangle_3\}$ be the minimum natural number that is divisible by $\Cyclelength$. 
\par
Finally, if $V(n) \leq L$, then $\|\Queuelength{}(n\triangle)\|_{\infty} \leq \sqrt{L}$. Also, $V(n+1) \leq V(n) + (\sqrt{L}+|\numsources|\tilde{\triangle})^2$ because the total number of arrivals at each time step does not exceed $|\numsources|$. Combining this with the previously considered case of $\Lyap{n}>L$, we get
\begin{equation*}
\begin{aligned}
    \mathbb{E}\left[\right.&\left.\Lyap{n+1} - \Lyap{n}|\: \StateofMCExp_{\tilde{\triangle}}(n)\right] \leq -\|\Queuelength{}(n\triangle)\|_{\infty} \\
    &+ \left((\sqrt{L}+|\numsources|\tilde{\triangle})^2 + \sqrt{L} + 1\right)\mathds{1}_{B},
\end{aligned}
\end{equation*} 
where $B = \{\StateofMCExp_{\tilde{\triangle}}(n) \in \StateSpace: V(n) \leq L\}$ (a finite set). The result then follows from the well-known Foster-Lyapunov drift criterion \cite[Theorem 14.0.1]{meyn2012markov}.
\subsection*{\underline{Proof of \eqref{eq:node-degree-inequality}}}
Consider \source{} $i \in \numsources$. If for every interval $[b_{i}m, b_{i}(m+1))$ in $[n\triangle, (n+1)\triangle-1]$, $a_{i}$ \agent{s} cross \source{} $i$, then 
\begin{equation*}
\begin{aligned}
        \Nodedegree{i}&((n+1)\triangle) \leq 
        \Nodedegree{i}(n\triangle) + \Numarrivals{i}(n\triangle+1,(n+1)\triangle) - \lfloor\frac{\triangle}{b_{i}}\rfloor a_{i} \\
        &\leq \Nodedegree{}(n\triangle) + \Numarrivals{i}(n\triangle+1,(n+1)\triangle) - \frac{a_{i}}{b_{i}}\triangle + a_{i}, 
\end{aligned}
\end{equation*}
where $\Numarrivals{i}(m_1,m_2)$ is the total number of arrivals in the network during the time interval $[m_1,m_2]$ that want to cross \source{} $i$'s node in order to reach their destination. If not, then there exist some time step in $[n\triangle, (n+1)\triangle-1]$ during which a valid slot that is empty crosses \source{} $i$'s node. Let $m_i \in [n\triangle, (n+1)\triangle-1]$ be the last time step at which this happens, i.e., $\Nodedegree{i}(m_i+1) = \Nodedegree{i}(m_i) + \Numarrivals{i}(m_i+1,m_i+2)$. Hence,
\begin{equation}\label{eq:node-degree-inequality-source-i}
\begin{aligned}
    \Nodedegree{i}((n+1)\triangle) &\leq \Nodedegree{i}(m_i+1) + \Numarrivals{i}(m_i+2,(n+1)\triangle) + a_{i} \\
    &- \frac{a_{i}}{b_{i}}((n+1)\triangle-m_i-1).
\end{aligned}
\end{equation}
\par
By the assumption of the theorem, i.e., $\Avgload{i} < \frac{a_{i}}{b_{i}}$ for all $i \in \numsources$, there exists $\delta > 0$ such that for all $i \in \numsources$,
\begin{equation*}
    \Avgload{i} + \delta <  \frac{a_{i}}{b_{i}}.
\end{equation*} 
\par
Thus, \eqref{eq:node-degree-inequality-source-i} can be rewritten as
\begin{equation}\label{eq:node-degree-inequality-source-i-rewritten}
\begin{aligned}
    \Nodedegree{i}((n+1)\triangle) &\leq \Nodedegree{i}(m_i+1) + \Numarrivals{i}(m_i+2,(n+1)\triangle)\\
    &- (\Avgload{i} + \delta)((n+1)\triangle-m_i-1) + a_{i}.
\end{aligned}
\end{equation}
\par
Furthermore, the queue size at \source{} $i$ at time $m_i+1$ cannot exceed $\Cyclelength$, i.e., $|\Queuelength{i}(m_i+1)| \leq \Cyclelength$. This is because, a valid slot that is empty crosses \source{} $i$ at time step $m_i$. Therefore, the quotas of \source{} $i$ at time $m_i$ must be zero. This implies that $\Queuelength{i}(m_i+1)$ only consists of \agent{s} arriving after the start of the most recent cycle before $m_i+1$. Since the number of arrivals to \source{} $i$ is no more than one per time step, it follows that $|\Queuelength{i}(m_i+1)| \leq \Cyclelength$. Hence,
        $\Nodedegree{i}(m_i+1) \leq \sum_{j \in \numsources_{i}^{-}} \Nodedegree{j}(m_i+1) + |\Queuelength{i}(m_i+1)| + \numcells{i} 
        \leq \sum_{j \in \numsources_{i}^{-}} \Nodedegree{j}(m_i+1) + \Cyclelength + \numcells{i}$,
where $\numcells{i}$ is the total number of slots between \source{} $i$ and any \source{} $j \in \numsources_{i}^{-}$.  
Combining this with \eqref{eq:node-degree-inequality-source-i-rewritten} gives
\comment{
\begin{equation*}
    \Nodedegree{i}(n+1) \leq \Nodedegree{i}(n-\triangle+1) - \Numdepartures{i}(m_i,k) + \Numarrivals{i}(k-\triangle+2,k+1) + \Cyclelength + a_{i},
\end{equation*}
where
\begin{equation*}
\begin{aligned}
    \Numdepartures{i}(s_i,k) &= (\Avgload{i}+\delta)(k-s_i+1) + \left(\Numarrivals{i}(k-\triangle+2,s_i+1) + \Nodedegree{i}(k-\triangle+1)\right)\indic{i}, \\
    \indic{i} &= \begin{cases}
    1 & \mbox{if } s_i > k-\triangle+1,~ I_{i}^{1}=\emptyset \\
    0 & \mbox{otherwise}
    \end{cases}.
\end{aligned}
\end{equation*}
}
\begin{equation}\label{Eq: (Prop 4) Starting inequality for off-ramp p node degree in terms of off-ramp p-1}
\begin{aligned}
    \Nodedegree{i}((n+1)&\triangle) \leq \sum_{j \in \numsources_{i}^{-}} \Nodedegree{j}(m_i+1) + \Numarrivals{i}(m_i+2,(n+1)\triangle) \\
    &- (\Avgload{i}+\delta)((n+1)\triangle-m_i-1) + \Cyclelength + \numcells{i} + a_{i}.
\end{aligned}
\end{equation}
\par
Repeating the above steps for any \source{} $j \in \numsources_{i}^{-}$ gives
\begin{equation*} 
\begin{aligned}
    \Nodedegree{j}(m_i+1) &\leq \sum_{k \in \numsources_{j}^{-}} \Nodedegree{k}(m_j+1) + \Numarrivals{j}(m_j+2,m_i+1) \\
    &- (\Avgload{j}+\delta)(m_i-m_j) + \Cyclelength + \numcells{j} + a_{j}, 
\end{aligned}
\end{equation*}
where $m_{j} \in [n\triangle, m_i+1]$ is defined similar to $m_{i}$ if at least one valid slot that is empty crosses \source{} $j$, and $m_{j} = n\triangle-1$ otherwise. This process can be repeated until we find a multiset $J_{i} \in U$ of \source{s}, such that for each $j \in J_i$ with multiplicity $\text{mult}(j)$, either $m_j^c = n\triangle-1$, or $m_j^c > n\triangle-1$ and $\numsources_{j}^{-} = \emptyset$, $c \in [\text{mult}(j)]$. Indeed, such a $J_{i}$ always exist since the network is acyclic. Let $\tilde{J}_{i}$ be a multiset containing all the \source{s} visited in this process (for example, if $J_{i} = \numsources_{i}^{-}$, then $\tilde{J}_{i} = \{i\}\cup \numsources_{i}^{-}$). For simplicity of notation, we drop the superscript $c$ in $m_j^c$ with the understanding that $m_j$ could take different values for \source{s} with multiplicity greater than one. We also use $j^{+}$ to denote an \source{} that succeeds \source{} $j$. By combining all the inequalities in this process we obtain
\begin{equation*}
\begin{aligned}
    \Nodedegree{i}((n+1)\triangle) &\leq \sum_{j \in J_{i}}\Nodedegree{j}(m_j+1) \\
    &+ \sum_{j \in \tilde{J}_{i}}\left(\Numarrivals{j}(m_{j}+2,m_{j^{+}}+1)\right. \\
    &\left.- (\Avgload{j}+\delta)(m_{j^{+}}-m_{j}) + \Cyclelength + \numcells{j} + a_{j} \right).
\end{aligned}
\end{equation*}
\par
Note that if $m_j>n\triangle-1$ for some $j \in J_{i}$, then $\Nodedegree{j}(m_j+1) \leq \Cyclelength \leq \Nodedegree{j}(n\triangle) + \Cyclelength$, which implies $\sum_{j \in J_{i}}\Nodedegree{j}(m_j+1) \leq \sum_{j \in J_{i}}\Nodedegree{j}(n\triangle) + |\numsources||\numpaths|\Cyclelength$. Therefore,
    $\Nodedegree{i}((n+1)\triangle) \leq \Nodedegree{}(n\triangle) + \MaxNumArrival{i}{}{\triangle} + C_2$,
where 
\begin{equation}\label{eq:Af-tilde-def-ramp}
\begin{aligned}
    &\MaxNumArrival{i}{}{\triangle}= \\
    & \sum_{j \in \tilde{J}_{i}}\left(\Numarrivals{j}(m_{j}+2,m_{j^{+}}+1) - (\Avgload{j}+\delta)(m_{j^{+}}-m_{j}) \right),
    \end{aligned}
\end{equation}
where we have dropped the dependence of $\MaxNumArrival{i}{}{\triangle}$ on other parameters for brevity. Generally, if we start with any $I \in U$, we can repeat the above steps to obtain
        $\sum_{k \in I}\Nodedegree{k}((n+1)\triangle) \leq \Nodedegree{}(n\triangle) + \MaxNumArrival{I}{}{\triangle} + C_2$,
where $\MaxNumArrival{I}{}{\triangle}$ has the same form as \eqref{eq:Af-tilde-def-ramp}, with $\tilde{J}_{I}$ defined similarly. Hence, \eqref{eq:node-degree-inequality} follows with
\begin{equation}\label{eq:Af-tilde-def}
    \MaxNumArrival{}{}{\triangle} = \max_{I \in U} \MaxNumArrival{I}{}{\triangle}.
\end{equation}
\par
Note that for each $I \in U$, we may assume that there exists some $j \in J_{I}$ for which $m_j = n\triangle-1$. For if $m_j>n\triangle-1$ for all $j \in J_{I}$, then $\Nodedegree{j}(m_j+1) \leq \Cyclelength$  for all $j \in J_{I}$, which implies 
\begin{equation}\label{eq:exclude-I-inequality}
\sum_{k \in I}\Nodedegree{k}((n+1)\triangle) \leq C_1\triangle + C_2 < \Nodedegree{}(n\triangle) - \triangle \leq \Nodedegree{}((n+1)\triangle),   
\end{equation}
where in the first inequality we have used the following fact: the total number of arrivals to the network is bounded by $|\numsources|$ at each time step. Hence, for each $j \in \tilde{J}_{I}$, $\Numarrivals{j}(m_{j}+2,m_{j^{+}}+1) \leq (m_{j^{+}}-m_{j})|\numsources| \leq |\numsources|\triangle$. Also, in the third inequality we have used the following fact: $\Nodedegree{}(\cdot)$ decreases by at most one during each time step. Therefore, $\Nodedegree{}((n+1)\triangle) \geq \Nodedegree{}(n\triangle) - \triangle$. The inequality \eqref{eq:exclude-I-inequality} implies that we can exclude $I$ from the maximum in \eqref{eq:Af-tilde-def}. Thus, we may assume such $j \in \tilde{J}_{I}$ with $m_j = n\triangle-1$ exists.
\subsection*{\underline{Proof of \eqref{eq:Kolmogorov-implication}}} 
Given $I \in U$ with the corresponding multiset $J_I$, let $m_j=n\triangle-1$ for some $j \in J_{I}$. Hence, there exists a sequence of \source{s} $\{i_1, \ldots, i_K\} \subset \tilde{J}_{I}$ with $i_1=j$ and $i_K \in I$ such that for $k<K$, $i_k^{+}=i_{k+1}$. Therefore, $\sum_{k \in [K]}(m_{i_{k}^{+}}-m_{i_{k}}) = \triangle$. Consider the sequence $\{\Numarrivals{i_k}(m,m+1)\}_{m = n\triangle+1}^{\infty}$, where the indices $i_k \in \{i_1,\ldots,i_K\}$ are allowed to depend on time $m$. For a given $m \in [n\triangle+1, \infty)$, the term $\Numarrivals{i_k}(m,m+1)$ is independent of the other terms in the sequence, $\Avgload{i_k} = \E{\Numarrivals{i_k}(m,m+1)}$ is bounded, and $\sigma^2_{i_k} := \E{\Numarrivals{i_k}^2(m,m+1) - \Avgload{i_k}^2}$ is (uniformly) bounded for all $i_k \in \{i_1,\ldots,i_K\}$. As a result, $\lim_{\triangle \ra \infty}\sum_{m = n\triangle+1}^{(n+1)\triangle}(m-n\triangle)^{-2}\sigma^2_{i_k}$ is also bounded. From Kolmogorov's strong law of large numbers~\cite[Theorem 10.12]{folland1999real}, we have, with probability $1$,
\begin{equation*}
    \lim_{\triangle \ra \infty}\frac{1}{\triangle}\left(\sum_{m = n\triangle+1}^{(n+1)\triangle}\Numarrivals{i_k}(m,m+1) - \sum_{m = n\triangle+1}^{(n+1)\triangle}\Avgload{i_k}\right) = 0.
\end{equation*}
\par
Hence, with probability one, there exists $\triangle_1, \eps > 0$ such that $\eps < \delta$ and for $\triangle > \triangle_1$ we have
\begin{equation*}
    \sum_{m = n\triangle+1}^{(n+1)\triangle}\Numarrivals{i_k}(m,m+1) - \sum_{m = n\triangle+1}^{(n+1)\triangle}\Avgload{i_k} < \eps \triangle.
\end{equation*}
\par
This implies that, with probability one, for $\triangle > \triangle_1$ we have
\begin{equation*}
\begin{aligned}
    \sum_{k \in [K]}\left(\Numarrivals{i_k}(m_{i_k}+2,m_{i_k^{+}}+1) - (\Avgload{k}+\delta)(m_{i_{k}^{+}}\right.&\left.-m_{i_{k}})  \right)\\
    &< (\eps - \delta)\triangle.
\end{aligned}
\end{equation*}
\par
Furthermore, for any other $j \in \tilde{J}_{I}\setminus\{i_1,\ldots,i_K\}$, with probability one, there exists $M_j > 0$ such that $m_{j^{+}}-m_{j} > M_j$ implies 
    $\Numarrivals{j}(m_{j}+2,m_{j^{+}}+1) - (\Avgload{j}+\delta)(m_{j^{+}}-m_{j}) < (\eps-\delta)(m_{j^{+}}-m_{j})$,
and if $m_{j^{+}}-m_{j} \leq M_j$, then $\Numarrivals{j}(m_{j}+2,m_{j^{+}}+1) - (\Avgload{j}+\delta)(m_{j^{+}}-m_{j}) \leq (|\numsources|-\Avgload{j}-\delta)M_j$, since the total number of arrivals to the network is bounded by $|\numsources|$ at each time step. Therefore, with probability one, for $\triangle > \max\{\triangle_1, \triangle_2:= \frac{1}{\delta-\eps} \sum_{j} (|\numsources|-\Avgload{j}-\delta)M_j\}$ we have
\begin{equation}\label{eq:Atilde-negative-sign}
\begin{aligned}
    \frac{1}{\triangle}\sum_{j \in \tilde{J}_{I}}&\left(\Numarrivals{j}(m_{j}+2,m_{j^{+}}+1) - (\Avgload{j}+\delta)(m_{j^{+}}-m_{j}) \right) \\
    &< (\eps-\delta) + \frac{1}{\triangle}\sum_{j}(|\numsources|-\Avgload{j}-\delta)M_j < 0
\end{aligned}
\end{equation}
which in turn implies that, with probability one, $\limsup_{\triangle \ra \infty}\MaxNumArrival{I}{}{\triangle}/\triangle < 0$. It follows that, with probability one, $\limsup_{\triangle \ra \infty}\MaxNumArrival{}{}{\triangle}/\triangle < 0$.
\par
Finally,
\begin{equation}\label{eq:bound-case-i1}
    \left|\frac{\MaxNumArrival{I}{}{\triangle}}{\triangle}\right| 
    \leq \frac{1}{\triangle}\sum_{j \in \tilde{J}_{I}}(|\numsources|+\Avgload{j}+\delta)(m_{j^{+}}-m_{j}) \leq C_2, 
\end{equation}
where we have used $\Numarrivals{j}(m_{j}+2,m_{j^{+}}+1) \leq (m_{j^{+}}-m_{j})|\numsources|$ in the first inequality using the fact that the total number of arrivals to the network is bounded by $|\numsources|$ at each time step. Therefore, the sequences $\left\{\MaxNumArrival{}{}{\triangle}/\triangle\right\}_{\triangle=1}^{\infty}$ and $\left\{\left(\MaxNumArrival{}{}{\triangle}/\triangle\right)^2\right\}_{\triangle=1}^{\infty}$ are upper bounded by integrable functions. Fatou's lemma \cite{folland1999real} then implies
\begin{equation*}
\begin{aligned}
    \limsup_{\triangle \ra \infty} \E{\frac{\MaxNumArrival{}{}{\triangle}}{\triangle}} &\leq \E{\limsup_{\triangle \ra \infty} \frac{\MaxNumArrival{}{}{\triangle}}{\triangle}} < 0, \\
    \limsup_{\triangle \ra \infty} \E{\left(\frac{\MaxNumArrival{}{}{\triangle}}{\triangle}\right)^2} &\leq \E{\limsup_{\triangle \ra \infty} \left(\frac{\MaxNumArrival{}{}{\triangle}}{\triangle}\right)^2} < C_3,
\end{aligned}
\end{equation*}
where $C_3 > C_{2}^{2}$ is some constant. This in turn gives \eqref{eq:Kolmogorov-implication}.


\end{document}